\pdfoutput=1
\documentclass[a4paper,english]{amsart}

\usepackage[sc,osf]{mathpazo}
\DeclareTextFontCommand{\textsl}{\fontfamily{ppl}\fontshape{sl}\selectfont}
\usepackage[euler-digits]{eulervm}
\usepackage{microtype}

\usepackage{amssymb,amscd,amsthm}
\usepackage[foot]{amsaddr}
\usepackage{mathtools}
\usepackage[
pdftitle={Strategy Recovery for Stochastic Mean Payoff Games},
pdfauthor={Marcello Mamino},
pdfkeywords={},
colorlinks=true,
citecolor=blue,
linkcolor=blue,
urlcolor=red]{hyperref}

\title[Strategy Recovery]{Strategy Recovery for Stochastic Mean Payoff Games}
\author[M.~Mamino]{\lsstyle Marcello Mamino}
\address{Laboratoire d'Informatique de l'\'Ecole Polytechnique (\textsc{lix})\\
B\^atiment Alan Turing\\
1 rue Honor\'e d'Estienne d'Orves, Campus de l'\'Ecole Polytechnique,
91120 Palaiseau, France.}
\email{\mailto{mamino@lix.polytechnique.fr}}
\thanks{The author has received funding from the European Research Council
under the European Community's Seventh Framework Programme (FP7/2007-2013
Grant Agreement no.~257039).}

\date{\mydate{16}{vi}{2015}}

\makeatletter
\def\Hy@Warning#1{}
\makeatother

\makeatletter
\def\@setemails{%
\ifnum\theg@author > 1 
\mbox{{\itshape E-mail addresses}:\space}{\ttfamily\emails}. 
\else 
\mbox{{\itshape E-mail address}:\space}{\ttfamily\emails}. 
\fi%
}
\makeatother

\makeatletter
\def\ps@plain{\ps@empty
  \def\@oddfoot{\normalfont\normalsize \hfil\thepage\hfil}%
  \let\@evenfoot\@oddfoot}
\def\ps@firstpage{\ps@plain
  \def\@oddfoot{\normalfont\normalsize \hfil\thepage\hfil
     \global\topskip\normaltopskip}%
  \let\@evenfoot\@oddfoot
  \def\@oddhead{\@serieslogo\hss}%
  \let\@evenhead\@oddhead 
}
\def\ps@headings{\ps@empty
  \def\@evenhead{%
    \setTrue{runhead}%
    \normalfont\normalsize
    \rlap{\thepage}\hfil
    \textsc{\lsstyle\MakeLowercase{\shortauthors}\hfil}}%
  \def\@oddhead{%
    \setTrue{runhead}%
    \normalfont\normalsize \hfil
    \textsc{\lsstyle\MakeLowercase{\rightmark{}{}}}\hfil\llap{\thepage}}%
  \let\@mkboth\markboth
}
\makeatother
\def\ifinlinemath#1#2{#2}
\def\mathshift{$}
\catcode`$=13
\def${\ifinner\expandafter\mathshift\else\expandafter\myshift\fi}
\def\myshift#1${{\def\ifinlinemath##1##2{##1}\raisebox{0ex}[0ex][0ex]{\mathshift#1\mathshift}}}
\AtBeginDocument{\catcode`$=13}

\def\myleftparen{(}
\def\myrightparen{)}
\catcode`(=13
\catcode`)=13
\def({\ifmmode\ifinlinemath{}{\expandafter\left}\fi\myleftparen}
\def){\ifmmode\ifinlinemath{}{\expandafter\right}\fi\myrightparen}

\makeatletter
\let\oldsection\section
\def\newsection#1{\oldsection{\lsstyle #1}}
\def\newsectionr#1{\oldsection*{\lsstyle #1}}
\def\section{\@ifstar\newsectionr\newsection}
\makeatother

\makeatletter
\newtheorem{ghost@theorem}{}[]
\def\@maketheorem#1=#2;{
	\newtheorem{#1}[ghost@theorem]{#2}}
\def\maketheorem#1{
	\@for\@x:=#1\do{
		\expandafter\@maketheorem\@x;}}
\makeatother

\def\N{\mathbb N}

\def\intervalcc#1#2{[\mkern1mu#1,#2\mkern1mu]}
\def\intervalco#1#2{[\mkern1mu#1,#2\mkern-1mu)}

\let\oldexists\exists\def\exists{\oldexists\mkern 1mu}

\let\n\oldstylenums
\def\mydate#1#2#3{\hbox{\n{#1}$\cdot${\sc#2}$\cdot$\n{#3}}}

\def\-{\nobreakdash-\hspace{0pt}}
\def\U-{\raise0.2ex\hbox{-}}
\def\url#1{\href{#1}{url\nobreakdash---\texttt{#1}}}

\def\mailto#1{\href{mailto:#1}{\texttt{#1}}}
\def\margin#1{\hbox to 0pt{\hss#1\ }}
\def\eatspace#1{#1}
\makeatletter
\def\myitem#1{%
\hfil\break\margin{\textsc{\hbox to 1em{\hss #1\hss}-}}%
\def\@currentlabel{\textsc{#1}}\eatspace}
\makeatother

\theoremstyle{plain}
\maketheorem{%
	theorem=Theorem,%
	proposition=Proposition,%
	lemma=Lemma,%
	corollary=Corollary,%
	fact=Fact%
}
\theoremstyle{definition}
\maketheorem{%
	definition=Definition,%
	notation=Notation,%
	example=Example%
}
\theoremstyle{remark}
\maketheorem{%
	observation=Observation,%
	remark=Remark%
}

\parskip=\baselineskip

\makeatletter
\@for\theoremstyle:=definition,remark,plain\do{%
\expandafter\g@addto@macro\csname th@\theoremstyle\endcsname{%
\setlength\thm@preskip\parskip%
\setlength\thm@postskip{0pt}%
}}

\renewenvironment{proof}[1][\proofname]{\par
  \pushQED{\qed}%
  \normalfont \topsep=0pt
  \trivlist
  \item[\hskip\labelsep
        \itshape
    #1\@addpunct{.}]\ignorespaces
}{%
  \popQED\endtrivlist\@endpefalse
}

\def\section{\@startsection{section}{1}%
  \z@\z@{1sp}%
  {\normalfont\scshape\centering}}
\makeatother

\makeatletter
\edef\orig@output{\the\output}
\output{\setbox\@cclv\vbox{\unvbox\@cclv\vspace{0pt plus 20pt minus 20pt}}\orig@output}
\makeatother

\def\pmax{\textsc{Max}}
\def\pmin{\textsc{Min}}
\def\v{\mathbf{v}}
\def\Q{\mathbb{Q}}
\def\N{\mathbb{N}}
\def\trans#1#2{\raisebox{-0.5ex}{$\mathsurround=0pt\textstyle \xrightarrow{#1,#2}$}}

\begin{document}

\begin{abstract}
We prove that to find optimal positional strategies for stochastic mean payoff
games when the value of every state of the game is known, in general, is as
hard as solving such games tout court. This answers a question posed by Daniel
Andersson and Peter Bro Miltersen.
\end{abstract}

\maketitle

In this note, we consider perfect information $0$-sum stochastic games,
which, for short, we will just call {\it stochastic games}. For us, a
stochastic game is a finite directed graph whose vertices we call
{\it states} and whose edges we call {\it transitions}, multiple edges and
loops are allowed but no state can be a sink. To each state~$s$ is associated
an~{\it owner}~$o(s)$ which is one of the two players~\pmax\ and~\pmin.
Each transition~$s\trans Ap t$ has an {\it action}~$A$
and a {\it probability}~$p\in\Q\cap\intervalcc01$, with the condition that, for
each state~$s$, the probabilities of the transitions exiting~$s$
associated to the same action must sum to~$1$. We say that the action~$A$
is {\it available} at state~$s$ if one of the transitions exiting~$s$ is
associated to~$A$. Furthermore to each action~$A$ is associated a {\it
reward\/}~$r(A)\in\Q$.

A play of a stochastic game~$G$ begins in some state~$s_0$ and produces an
unending sequence of states~$\{s_i\}_{i\in\N}$ and
actions~$\{A_i\}_{i\in\N}$. At move~$i$, the owner of
the current state~$s_i$ chooses an action~$A_i$ among those available at~$s_i$,
then one of the transitions exiting~$s_i$ with action~$A_i$ is selected at
random according to their respective probabilities, and the next
state~$s_{i+1}$ is the destination of the chosen transition.
A play can be evaluated according to the $\beta$-discounted payoff
criterion
\[
\v_\beta(A_0,A_1\dotsc) = (1-\beta) \sum_{i=0}^{\infty} r(A_i)\beta^i
\]
for $\beta\in\intervalco01$. Or it can be evaluated according to the mean
payoff criterion
\[
\v_1(A_0,A_1\dotsc) = \liminf_{n\to\infty} \frac{1}{n+1} \sum_{i=0}^{n} r(A_i)
\]
The goal of~\pmax\ is to maximize the evaluation, that of~\pmin\ is to
minimize it.
It is known that for both criteria there are optimal strategies which are
{\it positional}~\cite{Gil57,LiLi69}, namely such that the action chosen at~$s_i$
depends only on the state~$s_i$ -- an not, for instance, on the preceding
states in the play, on~$i$, or on a random choice. Given two positional
strategies~$\sigma$ and~$\tau$ for~\pmax\ and~\pmin\ respectively, and
given~$\beta\in\intervalcc01$, we
denote~$\v_\beta(G,s_0,\sigma,\tau)$ the expected value of~$\v_\beta$ on all plays
generated by~$\sigma$ and~$\tau$ starting from~$s_0$. We
write~$\v_\beta(G,s_0)$ for~$\v_\beta(G,s_0,\sigma,\tau)$ with~$\sigma$
and~$\tau$ optimal. For basic information on stochastic games one may
refer to the book~\cite{Filar-Vrieze}.

Given a stochastic game with probabilities and rewards encoded in binary,
and a value of~$\beta$ also encoded in binary,
it makes sense to study the computational complexity of the task of
solving the game. Strategically solving a game, as defined in~\cite{AnMil09},
means to find a pair of optimal strategies. Quantitatively solving~$G$ means
to find~$\v_\beta(G,s)$ for all states~$s$. In general, the second task is
easier than the first. The {\it strategy recovery problem} is,
given the quantitative solution of a game, to produce a strategic
solution. It has been observed in~\cite{AnMil09} that this task can be performed
trivially in linear time for discounted payoff games, and also, but not
trivially, for terminal payoff and simple stochastic games,
hence it was asked whether the same could be done for
stochastic mean payoff games (this is, indeed, the only missing element to
complete Andersson and Miltersen's picture).
Our aim is to prove that the strategy recovery problem for stochastic mean
payoff games is as hard as it possibly can.

\begin{theorem}\label{th-main}
The strategy recovery problem for stochastic mean payoff games is
equivalent, modulo polynomial time Turing reductions, to the task of
strategically solving mean payoff games.
\end{theorem}

We will combine the reduction from stochastic mean payoff to discounted
payoff games proven in~\cite{AnMil09} with a new reduction from discounted to
mean payoff games of a special form that we call~$\beta$-recurrent. Then
we will show that $\beta$-recurrent mean payoff games can be turned into
{\it strategically equivalent} mean payoff games having the additional
property that all states have value~$0$. For this latter class of games,
the strategy recovery problem is obviously equivalent to solving the games
strategically.

\begin{definition}
Let $G$ be a stochastic game and $s_0$ one of the states of~$G$.  We
define the {\it $\beta$-recurrent game associated to~$G$
and~$s_0$}, denoted~$G_{\beta,s_0}$.  The game~$G_{\beta,s_0}$ has the
same state-space as~$G$.  Each transition~$a\trans Ap b$ in~$G$ is
replaced by two new transitions in~$G_{\beta,s_0}$,
namely~$a\trans A{\beta p} b$ and~$a\trans A {(1-\beta)p} s_0$.
The first of these new transitions will be called~{\it of the first kind},
the second~{\it of the second kind}. We say that a game is {\it
$\beta$-recurrent} if it results from the construction just defined, for
some~$G$.
\end{definition}

Notice that our $\beta$-recurrent games are ergodic in the
sense of~\cite{BEGM10}. The complexity of ergodic games has been settled in a
recent work~\cite{ChaIJ14} (see the full version~\cite{ChaIJ14-full}),
however we need for our reduction the extra
properties of $\beta$-recurrent games. Interestingly, the definition of
ergodic in~\cite{ChaIJ14} is more restrictive than that in~\cite{BEGM10}, and, in
particular, in this stronger sense, a $\beta$-recurrent game may not be
ergodic, nor an ergodic game needs to be $\beta$-recurrent.

\begin{lemma}\label{th-lemma1}
The task of quantitatively solving stochastic discounted payoff games is
polynomial time Turing reducible to quantitatively solving
$\beta$-recurrent stochastic mean payoff games.
\end{lemma}
\begin{proof}
Consider a stochastic game~$G$ and discount factor~$\beta$. Let $s_0$
denote a state of~$G$. We will show that
\[
\v_\beta(G,s_0) = \v_1(G_{\beta,s_0},s_0)
\]
Intuitively, an infinite play of~$G_{\beta,s_0}$ can be seen as a sequence
of finite sub-plays, each of which lasts until a transition of the second
kind is taken and the game is reset to the initial state~$s_0$. Each
sub-play lasts at least one move, but a second move is played only with
probability~$\beta$, a third one with probability~$\beta^2$, and so on,
thus imitating the discounted payoff situation.

In order to prove the proposition, it suffices to show that, for any pair of
positional strategies~$\sigma$ and~$\tau$ for~\pmax\ and~\pmin\ respectively,
one has
\begin{equation*}\tag{$\star$}
\v_1(G_{\beta,s_0},s_0,\sigma,\tau) = \v_\beta(G,s_0,\sigma,\tau)
\end{equation*}
In fact, it follows from this equation that $\sigma$ and~$\tau$ are a pair
of optimal positional strategies for~$G_{\beta,s_0}$ if and only if they
are a pair of optimal positional strategies for~$G$ with starting
position~$s_0$.

It remains to prove equation~$(\star)$.
For each state~$s$ of~$G$, call~$A_{\sigma,\tau}(s)$ the action chosen by
either~$\sigma$ or~$\tau$ (according to the owner of~$s$) at the
state~$s$. The $\beta$-discounted values
of the states of~$G$ are determined by the condition
\[
\v_\beta(G,s,\sigma,\tau) = (1-\beta) r( A_{\sigma,\tau}(s) ) + \sum_{t\in G}
\beta p_{\sigma,\tau}(s\to t) \v_\beta(G,t,\sigma,\tau)
\]
where $p_{\sigma,\tau}(v\to w)$ denotes the probability that, from
state~$s$, a transition to state~$t$ is chosen when playing
strategy~$\sigma$ against~$\tau$. If we call $s_0\dotsc s_n$ the states
of~$G$ and $\bar v_\beta =
(\v_\beta(G,s_i,\sigma,\tau))_{i=1\dotsc n}$ the value
vector of~$G$, then the condition above can be rewritten in the form
\[
\bar v_\beta = (1-\beta) \bar r + \beta P \bar v_\beta
\]
where $\bar r$ is the vector of the rewards~$\bar r_i = r(
A_{\sigma,\tau}(s_i))$, and $P$
denotes the matrix of the transition probabilities~$P_{i,j} = p_{\sigma,\tau}(s_i \to s_j)$.
Hence
\[
\bar v_\beta = (1-\beta) (I - \beta P)^{-1} \bar r
\]
where $I$ denotes the $n\times n$ identity matrix.

Now we turn our attention to the mean payoff of the pair of
strategies~$\sigma$ and~$\tau$ in~$G_{\beta,s_0}$. We can compute
$\v_1(G_{\beta,s_0},s_0,\sigma,\tau)$
averaging the rewards over the stable distribution of the Markov chain
induced by these strategies on the states of~$G$. This stable
distribution~$\mu$ must be unique, because,
by virtue of $G_{\beta,s_0}$ being $\beta$-recurrent,
the Markov chain is connected. Moreover $\mu$ is determined by the condition
\[
\mu(s) = (1-\beta)\delta_{s_0}(s) + \sum_{t\in G} \beta p_{\sigma,\tau}(t,s)\mu(t)
\]
where $\delta_{s_0}(s)$ is~$1$ if~$s=s_0$ and~$0$ otherwise. Rewriting as
above, we get
\[
\bar \mu = (1-\beta)e_0 + \beta P^T \bar\mu
\]
where $e_0$ is the first element of the canonical basis and $\bar \mu_i =
\mu(s_i)$. Hence
\[
\bar \mu = (1-\beta) (I - \beta P^T)^{-1}  e_0
\]
Now, computing the average
\begin{align*}
\v_1(G_{\beta,s_0},s_0,\sigma,\tau) &= \sum_{s\in G} \mu(s) r(
A_{\sigma,\tau}(s)) \\
&= {\bar \mu}^T \bar r \\
&= e_0^T (1-\beta) (I - \beta P)^{-1} \bar r \\
&= e_0^T \bar v_\beta \\
&= \v_\beta(G,s_0,\sigma,\tau) \qedhere
\end{align*}
\end{proof}

\begin{lemma}\label{th-lemma2}
The task of strategically solving $\beta$-recurrent stochastic mean payoff games
is polynomial time many-one reducible to the strategy recovery problem for
stochastic mean payoff games.
\end{lemma}
\begin{proof}
Let $G_{\beta,s_0}$ be a $\beta$-recurrent stochastic game. As we noticed,
all the states of~$G_{\beta,s_0}$ have the same value. Nevertheless, we
have no obvious way to determine this value in order to complete the
reduction. Instead, we choose to construct a new mean payoff game~$G'$ in
such a way that all the states of~$G'$ get mean payoff value equal to~$0$,
and nonetheless a pair of optimal strategies for~$G_{\beta,s_0}$ can be
recovered from a pair of optimal strategies for~$G'$. This is clearly
sufficient to establish the lemma.

The game~$G'$ is constructed as two chained copies~$G^1$ and~$G^2$
of~$G_{\beta,s_0}$, redirecting all the transitions of the second kind in
each instance -- that go to the state corresponding to~$s_0$ in that
instance -- to the $s_0$-state in the other. The states of~$G^1$ have the
same owner as in~$G_{\beta,s_0}$, and the transitions originating in~$G^1$ are
associated to the same actions with the same rewards as in~$G_{\beta,s_0}$. In~$G_2$,
however, the owners are switched and the signs of the rewards exchanged
(formally we replace each action~$A$ with a new one~$A'$
having~$r(A')=-r(A)$).
If both players play optimally, we may expect each to win in~$G_1$ precisely
as much as he loses in~$G_2$, hence, arguably the value of~$G'$ should
be~$0$. On the other hand, in order to play optimally in~$G'$, one should
play optimally in both the components, so we should be able to extract
optimal positional strategies for~$G_{\beta,s_0}$ from optimal positional
strategies for~$G'$ by mere restriction to the component~$G^1$. We will
now proceed to prove our statement.

Let us denote by~$s^1$ and~$s^2$ respectively the states of~$G^1$
and~$G^2$ corresponding to a given state~$s$ of~$G_{\beta,s_0}$.
First observe that a play of~$G'$, almost surely, will eventually reach
state~$s_0^1$, from this follows that all the states of~$G'$ must have the
same value ($G'$~is ergodic).
A positional strategy~$\sigma$ for~\pmax\ in~$G'$ can be seen as a pair of
positional strategies~$(\sigma^1,\sigma^2)$ where~$\sigma^1$ is the strategy
for~\pmax\ in~$G_{\beta,s_0}$ that we get restricting~$\sigma$ to~$G^1$,
and $\sigma^2$ is the strategy for~\pmin\ in~$G_{\beta,s_0}$ that we get
from the restriction of~$\sigma$ to~$G^2$ (remember that in~$G^2$ the
players are switched).  Similarly a strategy~$\tau$ for~\pmin\ in~$G'$
can be seen as a pair of strategies~$(\tau^1,\tau^2)$ in~$G_{\beta,s_0}$,
the first one for~\pmin\ and the second for~\pmax. We will prove that for
any~$\sigma$ and~$\tau$
\begin{equation*}\tag{$\star\star$}
\v_1(G',\cdot,\sigma,\tau) =
\frac{1}{2} \v_1(G,\cdot,\sigma^1,\tau^1)
-
\frac{1}{2} \v_1(G,\cdot,\tau^2,\sigma^2)
\end{equation*}
From this equation, it follows at once that $\sigma$ is an optimal
strategy for~$G'$ if and only if $(\sigma^1,\sigma^2)$ is a pair of
optimal strategies for~$G_{\beta,s_0}$, and, in particular, the value of~$G'$ is~$0$.

We turn now to the proof of equation~($\star\star$). Consider the unique
stable distribution~$\mu$ of the Markov process induced by~$\sigma$
and~$\tau$. Observe that, independently from~$\sigma$ and~$\tau$, at any
given state, our Markov chain has probability~$\beta$ of transitioning to
a state belonging to the same component, and probability~$1-\beta$ of
switching component. It follows that the sequence of the components must
obey the law of a two-state Markov chain with transition matrix
\[
(
\begin{matrix}
\beta & 1-\beta \\
1-\beta & \beta
\end{matrix}
)
\]
Hence~$\mu(G^1)=\mu(G^2)=1/2$.
It suffices to
prove that the probability distributions~$\mu^1$ and~$\mu^2$ defined on
the states of~$G_{\beta,s_0}$ by
$\mu^1(s) = 2 \mu(s^1)$ and $\mu^2 = 2 \mu(s^2)$ are the stable distributions
induced on~$G_{\beta,s_0}$ by the pairs of strategies~$(\sigma^1,\tau^1)$
and~$(\tau^2,\sigma^2)$ respectively.

By symmetry, we can concentrate on~$\mu^1$.
Let $p_{\sigma,\tau}(t,s)$ denote the probability of the
transition~$t\to s$ in the Markov process
induced by the strategies~$\sigma$ and~$\tau$.
Since all states of~$G^1$
except~$s_0^1$ are only reachable from within~$G^1$ itself, the
consistency equation for~$\mu$ being a stable distribution on~$G'$
\[
\mu(s) = \sum_{t\in G'} p_{\sigma,\tau}(t,s)\mu(t)
\]
implies the same condition for~$\mu^1$ at all states except~$s_0$.
At~$s_0$ one concludes by direct computation observing that the component
of the sum on the right hand side due to transitions of the second kind
must be
\[(1-\beta)\mu(G^2) = \frac{1-\beta}{2} = (1-\beta)\mu(G^1) \qedhere\]
\end{proof}

\begin{proof}[Proof of Theorem~\ref{th-main}]
By \cite[Theorem~1]{AnMil09}, solving stochastic mean payoff games strategically
is reducible to solving stochastic discounted payoff games quantitatively,
which reduces, by Lemma~\ref{th-lemma1}, to solving $\beta$-recurrent
stochastic mean payoff games quantitatively. In turn, solving such
$\beta$-recurrent games quantitatively is reducible to solving the same
strategically, just because they are, in particular, stochastic mean
payoff games. By Lemma~\ref{th-lemma2}, this final task is reducible to
the strategy recovery problem for stochastic mean payoff games.
\end{proof}

Finally, we would like to remark that our construction relies on the
interpretation of {\it strategic solution} as requiring optimal {\it
positional} strategies. Were a more general class of strategies available,
then the problem of finding an optimal one would become easier. In
particular, the games produced by Lemma~\ref{th-lemma2} happen to be
symmetric under switching the players and the signs of the rewards. Under
this circumstance, it would not be surprising if one could play optimally
by some form of strategy stealing technique.

\section*{Acknowledgements}

We would like to express gratitude to Manuel Bodirsky and Eleonora
Bardelli for interesting discussions.


\end{document}